\documentclass{article}

\usepackage{amssymb, amsmath, amsopn}
\usepackage{times,mathptmx}
\usepackage{algorithm}
\usepackage[algo2e,ruled,vlined]{algorithm2e}
\usepackage{xspace}
\usepackage[numbers,sort&compress,longnamesfirst,sectionbib]{natbib}
\usepackage{balance}

\newtheorem{theorem}{Theorem}
\newtheorem{lemma}[theorem]{Lemma}

\newcommand{\ignore}[1]{}

\newcommand{\expect}[1]{\mathord{E}\mathord{\left[#1\right]}}

\newcommand{\ie}{i.\,e.\xspace}
\newcommand{\eg}{e.\,g.\xspace}
\newcommand{\wlo}{w.\,l.\,o.\,g.\xspace}
\newcommand{\Wlog}{W.\,l.\,o.\,g.\xspace}
\newcommand{\wrt}{w.\,r.\,t.\xspace}

\newcommand{\ea}{(1+1)~EA\xspace}
\newcommand{\EA}{(1+1)~EA\xspace}
\newcommand{\rls}{RLS\xspace}

\DeclareMathOperator{\Prob}{Pr}
\DeclareMathOperator{\Var}{Var}
\newcommand{\indic}[1]{\ensuremath{{\mathbf 1}}\{#1\}}

\renewcommand{\epsilon}{\varepsilon}

\makeatletter
\DeclareRobustCommand{\qed}{\ifmmode\mathqed\else\leavevmode\unskip\penalty9999\hbox{}\nobreak\hfill\quad\hbox{\qedsymbol}\fi}
\let\QED@stack\@empty
\let\qed@elt\relax
\newcommand{\pushQED}[1]{\toks@{\qed@elt{#1}}\@temptokena\expandafter{\QED@stack}\xdef\QED@stack{\the\toks@\the\@temptokena}}
\newcommand{\popQED}{\begingroup\let\qed@elt\popQED@elt \QED@stack\relax\relax\endgroup}
\def\popQED@elt#1#2\relax{#1\gdef\QED@stack{#2}}
\newcommand{\qedhere}{\begingroup \let\mathqed\math@qedhere\let\qed@elt\setQED@elt \QED@stack\relax\relax \endgroup}
\newif\ifmeasuring@
\newif\iffirstchoice@ \firstchoice@true
\def\setQED@elt#1#2\relax{\ifmeasuring@\else \iffirstchoice@ \gdef\QED@stack{\qed@elt{}#2}\fi\fi#1}
\newcommand{\mathqed}{\quad\hbox{\qedsymbol}}
\def\linebox@qed{\hfil\hbox{\qedsymbol}\hfilneg}
\def\math@qedhere{\@ifundefined{\@currenvir @qed}{\qed@warning\quad\hbox{\qedsymbol}}{\@xp\aftergroup\csname\@currenvir @qed\endcsname}}
\def\displaymath@qed{\relax\ifmmode\ifinner\aftergroup\linebox@qed\else\eqno\let\eqno\relax \let\leqno\relax \let\veqno\relax\hbox{\qedsymbol}\fi\else\aftergroup\linebox@qed\fi}
\@xp\let\csname equation*@qed\endcsname\displaymath@qed
\def\equation@qed{
  \iftagsleft@\hbox{\phantom{\quad\qedsymbol}}\gdef\alt@tag{\rlap{\hbox to\displaywidth{\hfil\qedsymbol}}\global\let\alt@tag\@empty}
  \else\gdef\alt@tag{\global\let\alt@tag\@empty\vtop{\ialign{\hfil####\cr\tagform@\theequation\cr\qedsymbol\cr}}\setbox\z@}
  \fi
}
\def\qed@tag{\global\tag@true \nonumber&\omit\setboxz@h {\strut@ \qedsymbol}\tagsleft@false\place@tag@gather\kern-\tabskip\ifst@rred \else \global\@eqnswtrue \fi \global\advance\row@\@ne \cr}
\def\split@qed{\def\endsplit{\crcr\egroup \egroup \ctagsplit@false \rendsplit@\aftergroup\align@qed}}
\def\align@qed{\ifmeasuring@ \tag*{\qedsymbol}\else \let\math@cr@@@\qed@tag\fi}
\@xp\let\csname align*@qed\endcsname\align@qed
\@xp\let\csname gather*@qed\endcsname\align@qed
\def\@tempb#1 v#2.#3\@nil{#2}
\ifnum\@xp\@xp\@xp\@tempb\csname ver@amsmath.sty\endcsname v0.0\@nil<\tw@\def\@tempa{TT}\else\def\@tempa{TF}\fi
\if\@tempa\renewcommand{\math@qedhere}{\quad\hbox{\qedsymbol}}\fi
\newcommand{\openbox}{\leavevmode\hbox to.77778em{\hfil\vrule\vbox to.675em{\hrule width.6em\vfil\hrule}\vrule\hfil}}
\DeclareRobustCommand{\textsquare}{\begingroup\usefont{U}{msa}{m}{n}\thr@@\endgroup}
\providecommand{\qedsymbol}{\openbox}
\newenvironment{proof}[1][\proofname]{\par\pushQED{\qed}\normalfont\topsep6\p@\@plus6\p@\relax\trivlist\item[\hskip\labelsep\itshape #1\@addpunct{.}]\ignorespaces}{\popQED\endtrivlist\@endpefalse}
\providecommand{\proofname}{Proof}
\makeatother

\begin{document}

\title{On the Runtime of Randomized Local Search and Simple Evolutionary Algorithms for Dynamic Makespan Scheduling}
\author{
 Frank Neumann\\
Optimisation and Logistics\\
School of Computer Science\\
The University of Adelaide\\
Adelaide, Australia
\and 
Carsten Witt\\
DTU Compute\\
Technical University of Denmark\\
2800 Kgs. Lyngby, Denmark
}

\maketitle
\begin{abstract}
Evolutionary algorithms have been frequently used for dynamic optimization problems. With this paper, we contribute to the theoretical understanding of this research area. We present the first computational complexity analysis of evolutionary algorithms for a dynamic variant of a classical combinatorial optimization problem, namely makespan scheduling.
We study the model of a strong adversary which is allowed to change one job at regular intervals. Furthermore, we investigate the setting of random changes.
 Our results show that randomized local search and a simple evolutionary algorithm are very effective in dynamically tracking changes made to the problem instance.

\end{abstract}

Optimization problems in real-world applications often change due to a changing environment. Evolutionary algorithms, ant colony optimization and other bio-inspired search heuristics have been frequently applied to dynamically changing problems.

An important approach to gain a theoretical understanding of evolutionary algorithms and other types of bio-inspired computation methods is the computational complexity analysis of these algorithms. During the last $20$ years, a large body of results and methods has been built up.
This includes the development of methods for the analysis of bio-inspired computing~\cite{DJWoneone,DBLP:journals/ai/HeY01,DBLP:journals/ai/HeY03,DBLP:conf/aaai/YuZ06}
and results for some of the best-known combinatorial optimization problems such as the traveling salesperson problem~\cite{Sutton2012tsp}, set cover~\cite{FriedrichHHNW10,DBLP:conf/ijcai/YuYZ13}, and makespan scheduling~\cite{Witt05,DBLP:conf/ppsn/SuttonN12} as well as different multi-objective problems~\cite{Neu07EJOR,DBLP:journals/ec/Horoba10,DBLP:journals/ai/QianYZ13}.
These studies often consider the algorithms called Randomized Local Search (\rls) and \EA, which we also investigate in this paper. Although these algorithms seem to be relatively simple, it should be noted that upper bounds on the expected optimization time of these algorithms can often be translated to population-based evolutionary algorithms with more complicated variation operators, \eg, crossover by increasing the upper bounds by only a linear factor with respect to population and problem size~\cite{DBLP:journals/ai/HeY03}.
 We refer the reader to
\cite{Neumann2010,Auger11,Jansen13} for comprehensive presentations of this research area. 

 In recent years, the computational complexity analysis of these algorithms on dynamically changing problems has gained increasing interest~\cite{RohlfshagenLehreYaoGECCO09,KotzingM12,OlivetoZargesGECCO13,LissovoiW14,LissovoiWTCS15,JansenZargesGECCO14}.
We study one of the classical combinatorial optimization problems, namely makespan scheduling on two machines.
 We consider \rls and \EA and analyze how they are able to keep track of changes that occur to the processing times of the given jobs.
In our investigations, we examine two models of dynamic changes where in each iteration at most the processing time of one job can be changed. In the adversary model, an adversary is able to change the processing time $p_i \in [L, U]$ of an arbitrary job $i$, possibly repeated at regular intervals. 
First, we show that even for very frequent and arbitrary changes, the algorithms are able to obtain solutions of discrepancy at most $U$ frequently during the run of the algorithm.
Afterwards, we show that \rls and \EA can maintain solutions of discrepancy at most~$U$ if the period of changes is not too small. 
In the random model, processing times are from the set $\{1, \ldots, n\}$ and an adversary is able to pick the job $i$ to be changed. The processing time $p_i$ of the chosen job is undergoing a random change and is either increased or decreased by $1$.
For the random model, we show that the \EA obtains solutions of discrepancy $O(\log n)$ in time $O(n^4 \log n)$ regardless of the initial solution and that the expected ratio between discrepancy and makespan is at most $6/n$ at least once in a phase of $O(n^{3/2})$ iterations.

The outline of the paper is as follows. We introduce the dynamic makespan problem and the algorithms under investigation in Section~\ref{sec:background}. Our analyses for the adversary model is presented in Section~\ref{sec:strong} and the random model is investigated in Section~\ref{sec:random}. Finally, we finish with some conclusions.

\section{Preliminaries}\label{sec:background}
We investigate the performance of randomized local search and a simple evolutionary algorithm for a dynamic version of the classical makespan problem. 
Given $n$ jobs and their processing times $p_i > 0$, $1 \leq i \leq n$, the goal is to assign each job to one of two machines $M_1$ and $M_2$ such that the makespan is minimized. 
A candidate solution is given by a vector $x \in \{0,1\}^n$, where job $i$ is assigned to machine~$M_1$ if $x_i=0$ and assigned to machine~$M_2$ if $x_i=1$, $1 \leq i \leq n$.

The makespan of a candidate solution $x$ is given by \[
f(x) = \max \left \{ \sum_{i=1}^n p_i (1- x_i), \sum_{i=1}^n p_i x_i \right \}
\]
and the goal is to find a solution $x^*$ which minimizes $f$. We denote by $|M_j|$ the load of machine $j=1,2$. We consider the dynamic version of the problem where exactly one job changes. 
We will also allow such changes to be repeated at regular intervals. 
We assume $p_i \in [L, U]$, $1 \leq i \leq n$, where $L$ is a lower bound on the processing time of any job and $U$ is an upper bound. We denote by $R = U/L$ the ratio between upper and lower bound.

\SetAlgoSkip{medskip}
\begin{algorithm2e}[t,h]
  \SetKwFor{For}{repeat}{}{}
  choose $x \in \{0,1\}^n$\;
  \While{stopping criteria not fullfilled}{%
    $y \gets x$\;
    flip one bit of $y$ chosen uniformly at random\;
    \lIf{$f(y) \leq f(x)$}{$x \gets y$}
  }
\caption{\rls.}
 \label{alg:rls}
\end{algorithm2e}

\SetAlgoSkip{medskip}
\begin{algorithm2e}[t,h]
  \SetKwFor{For}{repeat}{}{}
  choose $x \in \{0,1\}^n$\;
   \While{stopping criteria not fullfilled}{%
    $y \gets x$\;
    flip each bit of $y$ independently with prob.\ $1/n$\;
    \lIf{$f(y) \leq f(x)$}{$x \gets y$}
  }
\caption{\EA.}
 \label{alg:ea}
\end{algorithm2e}

Randomized local search (\rls) (see Algorithm~\ref{alg:rls}) starts with a candidate solution $x$ and produces in each iteration a new solution $y$ by flipping one randomly chosen bit of $x$. 
\EA (see Algorithm~\ref{alg:ea}) works with a more flexible mutation operator which flips each bit with probability $1/n$. The two introduced algorithms are standard benchmark algorithms in the area of runtime analysis of evolutionary computation~\cite{Neumann2010,Auger11,Jansen13}. While evolutionary algorithms usually work with a larger population and potentially also a crossover operator, usually positive statements on \EA transfer to elitist population-based evolutionary algorithms by losing only a polynomial factor dependent on the problem and population size~\cite{DBLP:journals/ai/HeY03}. This holds for all results obtained in this paper as well as long as there is in each iteration an inverse polynomial probability of selecting each individual of the parent population, selection 
does not accept worsenings of the worst fitness value from the population, and only the variation operator of \EA is applied.

We study the runtime behaviour of \rls and \EA on the introduced dynamic makespan scheduling problem and their ability to obtain solutions of good discrepancy. For our theoretical investigations, we do not consider any stopping criteria and measure  runtime by the number of iterations of the while-loop to achieve a solution of desired quality. The expected number of iterations is referred to as the expected time to reach the desired goal. 
In our investigations, we denote by 
\[
d(x)=\left\lvert\left(\sum_{i=1}^n p_i (1- x_i)\right) -  \left(\sum_{i=1}^n p_i x_i\right)\right\rvert,
\] the discrepancy of the solution $x$. We will study the expected time, for different scenarios, until \rls and \EA have produced solutions of small discrepancy.

We state an important property on the number of jobs on the fuller machine (\ie, the heavier loaded machine, which determines the makespan), 
which can easily be derived by taking into account the upper ($U$) and lower ($L$) bound on the processing times.

\begin{itemize}
\item Every solution has at least $\lceil(P/2)/U \rceil \geq \lceil (n/2)(L/U) \rceil =  \lceil(n/2)\cdot R^{-1} \rceil$ jobs on the fuller machine.
\end{itemize}

\section{Adversary Model}
\label{sec:strong}

In this section, we consider the case of a strong adversary. In one change, the adversary is allowed to pick one job $i$ to be changed and is able to choose an arbitrary new processing time $p_i \in [L, U]$.

\subsection{Obtaining a discrepancy of at most $U$}

We start our analysis by presenting upper bounds for \rls and \EA to obtain a discrepancy of at most $U$ from any starting solution.

\subsubsection{RLS}

We first consider \rls and show that the algorithm obtains a solution of discrepancy at most $U$ in expected time $O(n \min\{\log n, \log R\})$. This bound holds independently of the initial solution and the number of changes made by the adversary. The only requirement is that the adversary makes at most one change at a time. The proof uses the fact that for \rls the number of jobs on the fuller machine does not increase until the fuller machine switches.

\begin{theorem}
\label{theo:rls-strong}
The expected time until \rls has obtained a solution $x$ with $d(x) \leq U$ is $O(n \min\{\log n, \log R\})$ independently of the initial solution and the number of changes made by the adversary.
\end{theorem}

\begin{proof}
We assume that we are starting with an arbitrary solution assigning the jobs to the two machines. 
Let $x$ be the current solution and consider in each point in time the fuller machine. The number of jobs on the fuller machine does not increase as this would lead to a larger discrepancy.

We claim that if the fuller machine switched (either by moving a single job or by a single change of the adversary) then a solution of discrepancy at most $U$ has been obtained in the step before and after the switch.
Note that moving one job to another machine changes the load on each machine by at most $U$ and that the adversary can change the load on each machine by at most $U-L$. So, the step switching the fuller machine (accepted or rejected) reduces the load on the fuller machine from $P/2 +\alpha$, where $P = \sum_{i=1}^n p_i$, to $P/2 - \beta$ where $\alpha + \beta \leq U$. This implies $\min\{\alpha, \beta\} \leq U/2$ and therefore a discrepancy of at most $U$ directly before and/or after the fuller machine has switched. Note, that such a step is only accepted by \rls iff $\beta \leq \alpha$ and that a discrepancy of at most $U$ has been obtained if the step is accepted. On the other hand, the case $\alpha < \beta$ which is rejected by \rls implies a discrepancy of at most $U$ before the switch.

The fuller machine has at least $\lceil(n/2)\cdot R^{-1} \rceil$ jobs. 
Let $k$ be the number of jobs on the fuller machine. Then the probability to reduce the number of jobs on the fuller machine is  $\frac{k}{n}$ and the expected waiting time for such a step is
 $n/k$. Summing up, the expected time to switch the fuller machine is at most
\[
\sum_{k= \max\{\lceil (n/2)\cdot R^{-1} \rceil, 1 \}}^n \frac{n}{k} 
\]
 
We have two cases. If $R \geq n/2$, the sum is at most $nH_n = O(n \log n)$, where $H_n$ is the $n$-th Harmnoic number. If $R < n/2$, the sum is
at most $n \ln n +1  - n \ln (n/(2R)) = O(n \log R)$. Altogether, after at most $O(n \min\{\log n, \log R\})$ steps a solution of discrepancy at most $U$ has been obtained.
\end{proof}

\subsubsection{(1+1) EA}

In Theorem~\ref{theo:rls-strong}, we exploited that accepted steps of RLS cannot increase the number of jobs on the fuller machines. In contrast, the \ea 
may move few big jobs from the fuller to the emptier machine and many small jobs the other way round if the accumulated effect of 
the step 
decreases the discrepancy. Such multiple-bit flips, which may increase the number of jobs on the fuller machine, arise in a similar 
way in the analysis of the \ea on linear functions, where they complicate the analysis considerably \cite{WittCPC13}. However, 
it is also known that the number of incorrectly set bits in the \ea (corresponding to the number of jobs on the fuller machine) 
has a drift towards~$0$. We are going to show that this drift leads in time $O(n^{3/2})$ to the situation that the fuller machine switches, 
which was analyzed in Theorem~\ref{theo:rls-strong}. We cannot show the bound $O(n\log n)$ using the advanced drift techniques from 
\cite{WittCPC13} since the dynamics of the job sizes do not allow us to use the potential function from the literature.

\begin{theorem}
\label{theo:ea-strong}
The expected time until the \ea has obtained a solution $x$ with $d(x) \leq U$ is $O(n^{3/2})$ 
independently of the initial solution and the number of changes made by the adversary.
\end{theorem}

\begin{proof}
We start with a given search point $x_0$, where the time index 
\wlo\ is~$0$. \Wlog, $M_1$ is the fuller machine \wrt~$x_0$. We write $\ell_t$ to denote 
the load of~$M_1$ after $t$~steps. Now, let $T$ denote the 
first point in time where $\ell_t \le P/2+U/2$. At this time, 
$M_1$ might still be the fuller machine, which implies a discrepancy at most~$U$. Only if $
\ell_T < P/2-U/2$, the discrepancy is greater than~$U$. Note that 
$\ell_{T-1}>P/2+U/2$ and each job size is at most~$U$. Each step resulting in $\ell_T<P/2-U/2$ 
must flip at least $2$~bits and 
 can be converted into a step resulting in $\ell_T\in[P/2-U/2,P/2+U/2]$ 
by conditioning on that a certain subset of bits do not flip. Note that the 
step defined by  the stopping time~$T$ may be required to flip already more than 
one bit to reach $\ell_{T}\le P/2+U/2$ or even no bits may flip at all if 
the adversary is responsible for reaching $\ell_{T}\le P/2+U/2$; in the latter case, 
already discrepancy at most~$U$ has been obtained.  Note also 
that flipping bits in addition to the ones required to reach $\ell_{T}\le P/2+U/2$
may result in a rejected step. If we condition on the step flipping as few additional bits 
as possible, we are guaranteed to enter the interval $[P/2-U/2,P/2+U/2]$ for the load 
of the fuller machine, resulting in an accepted step. 
The probability of not flipping a certain subset of bits is at least $(1-1/n)^n \ge e^{-2}$. Hence, if the step leading to  
time~$T$ flips more than the required bits, we repeat the 
following analysis and increase the expected time by a factor of at most~$e^2$. 

We denote by $N_1(x_t)$ the number of jobs on $M_1$ with respect to~$x_t$, the current 
search point after~$t$ steps.  
Based on this, we define the potential function
\[
d(x_t) := \begin{cases}
N_1(x_t) & \text{ if $t<T$}\\
0 & \text{ otherwise.}
\end{cases} 
\]
Hence, the potential function reflects the number of jobs on the fuller machine before time~$T$ and is set to~$0$ afterwards. 
As we have argued, the discrepancy at time~$T$ is at most $U$ with probabiblity at least $e^{-2}$.

The aim now is to bound $\expect{T}$, which is achieved by bounding  the expression
$
\expect{d(x_t)-d(x_{t+1}) \mid x_t;t<T}
$
from below and performing drift analysis. In what follows, we use the notation $X_t:=d(x_t)$. 

Since it is necessary to move at least one job from the fuller machine to change the $d$-value, which happens 
with probability at least~$1/(en)$ for each of these jobs, and 
each job on the emptier machine switches machine with probability at most $1/n$, we get 
the bound on the drift 
\begin{equation}
\expect{X_t-X_{t+1} \mid X_t; t<T}
\ge \frac{X_t}{en}\left(1-\frac{n-X_t}{n} \right) = \frac{X_t^2}{en^2},
\label{eq:bound-drift-xt}
\end{equation}
which is at least $1/(en^2)$. Hence, despite the fact that the number of jobs on the fuller may increase, 
its decreases in expectation. 
Since the maximal $d$-value is~$n$, we get $\expect{T}=O(n^3)$ by additive drift analysis~\cite{DBLP:journals/ai/HeY01}. However, 
the pessimistic process analyzed here has already been more closely investigated in the literature. 
It has been (apart from irrelevant differences in details) 
modeled by a process called PO-EA by \cite{JansenBrittleness}, which was recently 
revisited by \cite{ColinDF14GECCO14}. Using this analysis, the bound can be improved to $O(n^{3/2})$. 

In the following, we present a self-contained 
proof of the $O(n^{3/2})$ bound 
using a novel potential function that 
is easier to handle than the one proposed in the literature. 
Intuitively, our potential function exploits that the process mostly moves due to the variance 
of the one-step change (instead of the very small drift) in the regime $X_t\le \sqrt{n}$ 
whereas it is governed by the actual drift $\expect{X_t-X_{t+1}\mid X_t}$
when $X_t$ is above $\sqrt{n}$. 

For $x\ge 0$, let the potential function be 
\[
g(x) :=\begin{cases}
x(\ln(\sqrt{n})+2-\ln(x))  & \text{ if $x\le \sqrt{n}$,}\\
3\sqrt{n} - \frac{n}{x} & \text{otherwise}.
\end{cases}
\]
We note that $g(x)$ is monotone increasing and continuous on 
$[0,n]$. Moreover, the 
derivative satisfies 
\[
g'(x):=\frac{\mathrm{d}g}{\mathrm{d}x}  
=\begin{cases}
\ln(\sqrt{n})+1-\ln(x)  & \text{ if $x\le \sqrt{n}$,}\\
\frac{n}{x^2} & \text{otherwise}.
\end{cases}
\]
and is non-increasing and continuous as well. Hence, $g(x)$ is a concave function. The second derivative 
equals 
\[
g''(x):= \frac{\mathrm{d}^2 g}{\mathrm{d} x^2}  
=\begin{cases}
-1/x  & \text{ if $x\le \sqrt{n}$,}\\
-\frac{2n}{x^3} & \text{otherwise}.
\end{cases}
\]
and satisfies $g''(x)\le -1/x$ for $x\le n$.

By the mean-value 
theorem, we get for all $x\le n$ and 
for $y\ge 0$ that 
\begin{equation}
g(x) -g(x-y) \ge y g'(x) \ge g(x+y) - g(x).
\label{eq:g-first}
\end{equation}
Moreover, by developing Taylor expansions of $g(x-y)$ and 
$g(x+y)$ up to terms of fourth order, it is easy to see that 
\begin{equation}
(g(x) -g(x-y)) - (g(x+y) - g(x)) \ge 
-\frac{\mathrm{d}^2 g}{\mathrm{d} x^2} \ge \frac{1}{x}.
\label{eq:g-second}
\end{equation}

We are now going to analyze the drift of the process defined by $Y_t:=g(X_t)$. 
To this end, it is useful to decompose the drift 
 into a positive and negative part. 
Define  
\[
\Delta_X^-:=(X_t-X_{t+1})\cdot \indic{X_{t+1}\le X_t}
\]
and
\[
\Delta_X^+:=(X_{t+1}-X_{t})\cdot \indic{X_{t+1}\ge X_t}
\]
and accordingly $\Delta_Y^+$ and $\Delta_Y^-$ with respect to the $Y$-process. Note that 
$\expect{X_t-X_{t+1}\mid X_t} = \expect{\Delta_X^-\mid X_t} - \expect{\Delta_X^+\mid X_t}$ and 
accordingly for the drift of the $Y$-process.

Combining this decomposition with \eqref{eq:g-first}, we obtain 
\begin{align*}
& \expect{Y_t-Y_{t+1}\mid X_t} 
\\ 
& \ge g'(X_t)\cdot 
\expect{\Delta_X^-\mid X_t} - g'(X_t) \expect{\Delta_X^+\mid X_t} \\
&
=  g'(X_t) 
\expect{X_t-X_{t+1}\mid X_t} .
\end{align*}
If $X_t > \sqrt{n}$, plugging in the expression for $g'(X_t)$ and the bound 
\eqref{eq:bound-drift-xt}
yields 
\[
\expect{Y_t-Y_{t+1}\mid X_t}
\ge 
\frac{n}{X_t^2} \cdot \frac{X_t^2}{en^2} = \frac{1}{en},
\]
which does not depend on $X_t$. 

If $X_t\le \sqrt{n}$, we combine 
the decomposition with \eqref{eq:g-second} and get for some  value $a(X_t)$ that 
\begin{align*}
& \expect{Y_t-Y_{t+1} \mid X_t}  \\
& \quad \ge \left(a(X_t) + \frac{1}{X_t}\right)
\expect{\Delta_X^-\mid X_t} - a(X_t) \expect{\Delta_X^+ \mid X_t} \\
& \quad\ge \frac{\expect{\Delta_X^-\mid X_t}}{X_t} + a(X_t)
\expect{X_t-X_{t+1}\mid X_t}\\
&  \quad\ge \frac{\expect{\Delta_X^-\mid X_t}}{X_t}
\end{align*}
since $\expect{X_t-X_{t+1} \mid X_t}\ge 0$ according to \eqref{eq:bound-drift-xt}. 
Hence, we are left with a bound on $\expect{\Delta_X^-\mid X_t}$. 
Here we again argue that the number of jobs decreases by~$1$ if one of the $X_t$ jobs from the fuller 
machine moves 
and no other jobs moves. Consequently, $\expect{\Delta_X^-\mid X_t} \ge \frac{X_t}{en}$ and 
\begin{equation*}
\expect{Y_t-Y_{t+1} \mid X_t} \ge \frac{1}{en} 
\end{equation*} 
if   $X_t\le \sqrt{n}$. 
Together with the bound derived above, we have
$
\expect{Y_t-Y_{t+1} \mid X_t} \ge \frac{1}{en}
$
for every $X_t\le n$. Now, since $Y_0\le 3n^{1/2}$, the additive drift theorem
yields
$
\expect{T} \le 3en^{3/2} = O(n^{3/2})
$
as suggested.
\end{proof}

\subsection{Recovering a discrepancy of at most $U$}

We now consider the situation where the algorithms have obtained a solution of discrepancy at most $U$ and the processing time of one arbitrary 
job is changed afterwards. We show an upper bound of $O(\min\{R, n\})$ on the time needed to obtain a discrepancy of $U$ after this change.

\begin{theorem}
\label{thm:recompute}
Let $x$ be the current solution that has a discrepancy of at most $U$ before changing the processing time of a job on the fuller machine.
Then, the expected time of RLS and \ea to obtain a discrepancy of at most $U$ is $O(\min\{R, n\})$.
\end{theorem}

\begin{proof}
We use multiplicative drift analysis~\cite{DBLP:journals/algorithmica/DoerrJW12} to show the $O(n)$ bound and consider drift according to the discrepancy $d(x)$. Let 
$P = \sum_{i=1}^n p_i$ and $X_t$ be the random variable for $d(x)$ of the search point $x$ at time $t\geq 0$.
With respect to the parameters from the multiplicative drift theorem, we have $s_0 \leq U + (U-L)$, $s_{\min} = U$ and therefore $s_0 / s_{\min} \leq 2$. \Wlog, let $1, \ldots, f$ be the jobs on the fuller machine and  $p_1, \ldots p_f$ be their processing times. Furthermore let $y(i)$ be the search point obtained by flipping the bit $i$ for $i=1, \ldots, f$. As long as the current solution $x$ has discrepancy greater than $U$, each of these single bit flips is accepted. We get
\begin{align*}
& \expect{X_t - X_{t+1}\mid X_t}   \geq  \frac{1}{n} \cdot \left(1 - \frac{1}{n} \right)^{n-1} \cdot \sum_{i=1}^f ( d(x) - d(y(i))\\
& \quad \geq  \frac{1}{en} \left(\sum_{i=1}^{f} 2 \cdot p_i \right)
 \geq  \frac{2}{en}  (P/2 + d(x)/2))\\
& \quad =  \frac{1}{en}  (P + d(x)) \; \geq  \;\frac{1}{en}  d(x).
\end{align*}
We set  $\delta = 1/(en)$ and get $en \ln (s_0/s_{\min}) \leq en \ln 2 = O(n)$ as an upper bound.

It remains to show the $O(R)$ bound.
From the previous calculation, we already have 
\[
\expect{X_t - X_{t+1}\mid X_t} \geq \frac{1}{en}  (P + d(x)) \geq P/(en).
\]
Using additive drift analysis, the expected time to reach a discrepancy of at most $U$ when starting with a solution $x$ with $d(x) \leq U + (U-L)$ is
\[
(U-L) / (P / (en)) \leq en(U-L)/ (nL) = e(R-1) = O(R).
\]
Altogether the upper bound is $O(\min\{R, n\})$, which completes the proof.
\end{proof}

The previous theorem implies that both algorithms are effectively tracking solutions of discrepancy $O(U)$ 
if the time where no changes to the processing times are happening is at least $c \cdot \min\{R, n\}$, where $c$ is an appropriate constant.
In particular,  changes happening every $c'n$ iterations where $c'$ is an appropriate constant can be tracked 
effectively regardless of the ratio $R=U/L$. Furthermore, a small   ratio $R$, \eg\ a constant, implies that 
very frequent changes (every $c''R$ iterations, $c''$ an appropriate constant) can be tracked by \rls and \EA. These statements 
 can be obtained by combining drift analysis 
with an averaging argument over a number of phases. Due to space restrictions, this analysis 
is not spelt out here.

\section{Random Model}
\label{sec:random}

We now consider a model with less adversarial power. Dynamic changes are still possible, but each change 
is limited in effect. More precisely, we consider a random model as common in the average-case 
analysis of algorithms \cite{Witt05}. For simplicity, we consider the model where 
all jobs sizes are in $\{1,\dots,n\}$; generalizations to other sets 
are possible. At each point of time, at most one job size can be changed by the adversary. The adversary 
can only choose the job to change, but neither amount or direction of change. If a job~$i$ is chosen 
to change, then its processing time changes from its current value~$p_i$ to one of the 
two values $p_{i}+1$ and $p_{i}-1$, each with probability~$1/2$. Two exceptions are made if $p_i=n$, 
which results in job size $n-1$, and if $p_i=1$, which results in job size~$2$ afterwards. In other 
words, the size of each job performs a fair random walk on $\{1,\dots,n\}$, with reflecting 
barriers. With respect to the initial job sizes, we consider both arbitrary (worst-case) initializations 
and the case that the sizes are drawn uniformly at random and independently 
from $\{1,\dots,n\}$.
Then each initial job size is $(n+1)/2$ in 
expectation.

It is useful to denote the random processing time of job~$i$ at time~$t$ by the random variable~$X_i(t)$. It 
is well known \cite{LPWMarkov} that the random walk described by the process $X_i(t)$, $t\ge 0$, has 
a stationary probability distribution given by 
\[
\lim_{t\to\infty} \Prob(X_i(t)=j)=
\begin{cases}
1/(2n-2) & \text{if $j=1$ or $j=n$}\\
1/(n-1) & \text{otherwise }
\end{cases}
\]
Hence, the probability values in the
 stationary distribution differ from the initial uniform distribution by a factor of at most~$2$. It is also well known 
in the theory of random walks 
 that the so-called mixing time (informally, the time to get sufficiently close to the 
stationary distribution) of the considered random walk is $O(n^2)$ steps. Hence,  
for any~$i, j\in\{1,\dots,n\}$ and for any $t\ge cn^2$, where 
$t$ denotes the number of changes to job~$p_i$ and $c$ is a sufficiently large constant,  
we have 
\[
\frac{c_1}{n} \le \Prob(X_i(t)=j)\le  \frac{c_2}{n}
\]
for two constants $c_1,c_2>0$. Hereinafter, we asssume this bracketing 
of $X_i(t)$ to hold, \ie, the mixing time has elapsed for every job.

The aim is to analyze the discrepancies obtainable in our model. 
We summarize in the following lemma a useful property of the distribution of the 
processing times~$X_i(t)$, and drop the time index for convenience. 
Roughly speaking, it shows that there are no big gaps in the 
set of values that is taken by at least one job.

\begin{lemma}
\label{lem:random-gap}
Let $\phi(i):=\lvert\{X_j\mid X_j=i \wedge j\in\{1,\dots,n\}\}\vert$, where $i\in\{1,\dots,n\}$, be 
the frequency of jobs size $i$. Let $G:=\max\{\ell\mid \exists i\colon \phi(i)=\phi(i+1)=\dots=\phi(i+\ell)=0\}$ 
the maximum gap size, \ie maximum size of intervals with zero frequency everywhere. Then, for 
some constant $c>0$, 
\[
\Prob(G\ge \ell)\le n 2^{-c\ell}.
\]
\end{lemma}

\begin{proof}
Recall that we assume to be close to the stationary distribution, more precisely 
for each $i,j\in\{1,\dots,n\}$, we have 
$\Prob(X_j=i)\ge c_1/n$. By considering disjoint events, 
\[
\Prob(X_j \in\{i,\dots,i+\ell\})\ge \frac{c_1 \ell}{n}
\]
for each $\ell\le n$.

Then for each $\ell\ge 1$, we get from the 
independence of the job sizes that
\[
\Prob(\forall j\in\{1,\dots,n\}\colon X_j\notin \{i,\dots,i+\ell\}) \le \left(1-\frac{c_1 \ell}{n}\right)^n,
\]
which is at most  $c_3^\ell$ for some constant $c_3<1$. Hence, by a union bound the 
probability that there is an~$i$ such that 
for all $j\in\{1,\dots,n\}\colon X_j\notin \{i,\dots,i+\ell\}$ 
is at most $nc_3^\ell$, which 
equals
$n2^{-c\ell}$ for some constant~$c>0$.
\end{proof}

Hereinafter, \emph{with high probability} means 
probability at least $1-O(n^{-c})$ for 
any constant~$c>0$. We prove the following 
theorem stating that with high probability 
discrepancy $O(\log n)$ can be reached in polynomial time. Its proof is 
inspired by the average-case analysis from \cite{Witt05}. 
 Note also that the theorem 
is restricted to the \ea since its proof analyzes improvements made 
by swapping two jobs between the fuller and emptier machine.

\begin{theorem}
\label{theo:discrepancy-logn}
Regardless of the initial solution, the following claim holds: 
with high probability 
  the time for the \ea after a one-time change  
to obtain a discrepancy of at most $O(\log n)$ is $O(n^4 \log n)$.
\end{theorem}

\begin{proof}
According to Lemma~\ref{lem:random-gap}, there is for any constant~$c>0$ a sufficiently large constant 
$c'>0$ such that 
there is not gap of size $G:=c'\log n$ or larger with probabiilty 
at least $1-n\cdot n^{-c-1}=1-n^{-c}$. In the following, we assume this maximum gap size to hold. 

If the current discrepancy is larger than $G$, then there must be 
either   
 at least one pair of jobs $j,j'$ with $j$ on the fuller machine and 
$j'$ on the emptier machine such that $X_{j'}<X_j$ and $X_j-X_{j'} \le G$,
or a job $j$ on the fuller machine of size of at most~$G$. 
To see this, imagine that despite the gap size of at most $G$, there 
is no such pair as in the first case. Then all jobs of size at least $G$ 
must be on the fuller machine, resulting in the second case.

Now, in the first case it is sufficient to swap jobs $j$ and~$j'$ to decrease 
the discrepancy by at least~$1$. In the second case, it is enough to move job $j$ 
from the fuller to the emptier machine to decrease the discrepancy by at least~$1$. 
In any case, the probability of decreasing the discrepancy is 
at least $(1-1/n)^{n-1} \frac{1}{n^2}=\Omega(n^{-2})$. Since 
the maximum discrepancy is $O(n^2)$, the expected number of decreases 
is also at most $O(n^2)$. Multiplying this with the waiting time for an improvement, 
we have an expected time of $O(n^4)$. By a simple application of Markov's inequality 
and repeating phases of length $cn^4$ for some constant~$c$, it is easy to see that 
the time is $O(n^4\log n)$ with high probability.
\end{proof}

The previous theorem covers  a worst-case initialization with all jobs on one machine, where the discrepancy can be up to $n^2$. Under random 
initialization, this is unlikely to happen, as the following theorem shows.

\begin{theorem}
The expected discrepancy of the random initial solution is 
$\Theta(n\sqrt{n})$. Under a random initial solution, 
  the time for the \ea after 
	a one-time change 
to obtain a discrepancy of $O(\log n)$ is $O(n^{3.5} \log^2 n)$ 
with high probability.
\end{theorem}

\begin{proof}
We prove that the initial discrepancy is $\Theta(n\sqrt{n})$ in expectation 
and $O(n\sqrt{n}\log n)$ with high probability. From the last property, the 
statement on the time to obtain a discrepancy of $O(\log n)$ follows 
with the same ideas as in the proof of Theorem~\ref{theo:discrepancy-logn} if 
the initial discrepancy is estimated with $O(n\sqrt{n}\log n)$ 
instead of $O(n^2)$. 

We are left with the proofs on the initial discrepancy. Let $K$ 
denote the number of jobs that are initially put on the first machine. 
Then $\expect{K}=n/2$ but also $\expect{\lvert K-n/2\rvert} \le \sqrt{\Var(K)} = 
\Theta(\sqrt{n})$, where we used Jensen's inequality and the 
fact that $K\sim \text{Bin}(n,1/2)$. Moreover, by the properties 
of the binomial distribution we have that $\Prob(K\ge n/2+c\sqrt{n})=\Omega(1)$ 
for some constant~$c>0$. Altogether, $\expect{\lvert K-n/2\rvert} = \Theta(\sqrt{n})$. 
In other words, there are in expectation $\Theta(\sqrt{n})$ more jobs on one machine than 
on the other.

Each job size is initially uniformly distributed on $\{1,\dots,n\}$ and has 
expectation $(n+1)/2$. By linearity of expectation, the discrepancy 
is at least $\Theta(\sqrt{n})(n+1)/2=\Theta(n\sqrt{n})$. This consideration 
just subtracts the total load of the machine having the minority of the jobs from the total load 
of the machine having the majority (hereinafter called machine ``majority''). Should this difference be negative, the discrepancy is 
still positive. However, by approximating the sum of the job sizes on machine ``majority'' by 
a normal distribution with standard deviation $\Theta(n\sqrt{n})$, one can also 
see that the discrepancy is $O(n\sqrt{n})$ even if machine ``majority'' is allowed
to have a total load less than the other machine. Altogether the expected discrepancy is $\Theta(n\sqrt{n})$. 

The statement that the discrepancy is $O(n\sqrt{n}\log n)$ with high probability 
follows by using Chernoff bounds on the difference in the number of jobs between 
the two machines (stating that there are at most $O(\sqrt{n}\log n)$ more jobs 
on one machine than the other), approximating the tails of the sum of the job 
sizes  
on the machines by a normal distribution and arguing that a deviation 
of  $c\sqrt{n}\log n$ from the mean has probability $e^{-\Omega(c^2)}$. 
\end{proof}

Finally, we turn to the case that job sizes change frequently. In the extreme 
case, at every point of time one job size is allowed to increase or 
decrease by~$1$. Then it seems hard to obtain a discrepancy of $O(\log n)$ as 
shown in Theorem~\ref{theo:discrepancy-logn}. However, we can apply the results 
from Section~\ref{sec:strong}, noting that the maximum job size is~$n$ at any time. In relation 
to the makespan, the discrepancy, which will also be at most~$n$, is negligible.

\begin{theorem}
In the model with random changes, the following holds: 
the expected time until the \ea (\rls) has obtained 
 a solution with discrepancy at most~$n$ is $O(n^{3/2})$ (respectively 
$O(n\log n)$)  independently
of the initial solution and the number of changes.
The expected ratio between discrepancy and makespan is 
at most $6/n$ then. 
\end{theorem}

\begin{proof}
Since $R=U=n$ in the notation of  Theorem~\ref{theo:rls-strong} and Theorem~\ref{theo:ea-strong}, 
we immediately obtain the 
first statement of our theorem. To compute the expected ratio, note that at any time the sum 
of all job sizes has an expected value of $n(n+1)/2$ and is at least $n^2/3+n$ 
with probability $1-2^{-\Omega(n)}$ using the approximation by Normal distribution. 
In this case, the makespan must be at least $n^2/6+n/2$, and the ratio is 
at most $n/(n^2/6+n/2) \le 6/n-3/n$. If 
the sum of the job sizes is less than $n^2/3+n$, then the ratio is at most $n/n$ since 
all job sizes are at least one. Altogether, the expected ratio is bounded from above by 
$6/n-3/n + 2^{-\Omega(n)} \le 6/n$.
\end{proof}

\section{Conclusions}
We have shown that randomized local search and evolutionary algorithms are provably
 successful in tracking solutions of good discrepancy for the dynamic makespan scheduling. 
Investigating the adversary model, we have shown that the algorithms obtain solutions 
of discrepancy at most $U$ every $O(n \log n)$ (for \rls) and every $O(n^{3/2})$ (for \EA) iterations 
even if changes are arbitrary and frequent. Furthermore, such a discrepancy is maintained 
if the period of changes is not too small. For the random model, we have shown that discrepancies of 
$O(\log n)$ are obtained and that a ratio of at most $6/n$ between discrepancy and makespan is obtained frequently during the optimization process.

\section*{Acknowledgements}
Frank Neumann has been supported by
the Australian Research Council (ARC) through grants  DP130104395 and DP140103400.
Carsten Witt has been supported by the Danish Council for Independent Research (DFF) through grant 4002-00542. 
We thank Mojgan Pourhassan and three anonymous reviewers for providing valuable feedback that helped to improve the paper.

\end{document}